\documentclass[letterpaper, 10 pt, conference]{ieeeconf}  % Comment this line out if you need a4paper
\usepackage[font=footnotesize,labelfont=bf]{caption}

\IEEEoverridecommandlockouts                              % This command is only needed if 
\usepackage{graphicx} % for pdf, bitmapped graphics files
\usepackage{subcaption}
\usepackage{epsfig} % for postscript graphics files
\usepackage{mathptmx} % assumes new font selection scheme installed
\usepackage{times}

\usepackage{amsmath,amssymb,amsfonts, amsthm}
\usepackage{bm}
\usepackage{algorithmic}
\usepackage{yhmath}
\usepackage{textcomp}
\usepackage{xcolor}
\usepackage{cite}
\usepackage{url}

\newtheorem{theorem}{Theorem}
\newtheorem{lemma}{Lemma}

\newtheorem{assumption}{Assumption} % <-- defines an Assumption block
\theoremstyle{remark}

\title{\LARGE \bf
A Robust Neural Control Design for Multi-drone Slung Payload Manipulation with Control Contraction Metrics
}

\author{
  Xinyuan Liang$^{\dagger}$, Longhao Qian$^{\dagger}$, Yi Lok Lo$^{\dagger}$ and Hugh H.T. Liu$^{\dagger}$
  \thanks{$^{\dagger}$The authors are with Flight Systems and Control Lab, Institute for Aerospace Studies, University of Toronto, 4925 Dufferin St, North York, ON M3H 5T5, Canada {\tt\small \{xiny.liang, longhao.qian, enoch.lo\}@mail.utoronto.ca, liu@utias.utoronto.ca}}
}

\begin{document}

\maketitle
\thispagestyle{empty}
\pagestyle{empty}

%%%%%%%%%%%%%%%%%%%%%%%%%%%%%%%%%%%%%%%%%%%%%%%%%%%%%%%%%%%%%%%%%%%%%%%%%%%%%%%%
\begin{abstract}

This paper presents a robust neural control design for a three-drone slung payload transportation system to track a reference path under external disturbances. The control contraction metric (CCM) is used to generate a neural exponentially converging baseline controller while complying with control input saturation constraints. We also incorporate the uncertainty and disturbance estimator (UDE) technique to dynamically compensate for persistent disturbances. The proposed framework yields a modularized design, allowing the controller and estimator to perform their individual tasks and achieve a zero trajectory tracking error if the disturbances meet certain assumptions. The stability and robustness of the complete system, incorporating both the CCM controller and the UDE compensator, are presented. Simulations are conducted to demonstrate the capability of the proposed control design to follow complicated trajectories under external disturbances.

\end{abstract}

%%%%%%%%%%%%%%%%%%%%%%%%%%%%%%%%%%%%%%%%%%%%%%%%%%%%%%%%%%%%%%%%%%%%%%%%%%%%%%%%
\section{INTRODUCTION}

Modern developments in cable-suspended payload transportation using multirotors present various challenges related to system performance, stability, and safety. Ref.\cite{qian2019TOIE} proposed an uncertainty and disturbance estimator (UDE)-based technique for such a slung payload task using a single-drone design. However, compared to a single-agent slung payload system, a multi-drone design offers a more scalable solution with better range, higher payload capacity, additional redundancy, and provides improved localization accuracy thanks to increased sensor data \cite{maza2009multi}. Various improvements have been made for the proposed multi-drone payload scheme \cite{lee2017geometric, li2021cooperative, cai2024robust, costantini2025cooperative, wahba2024efficient, zhang2021self, goodman2023geometric, zhao2023composite}.

It is difficult to prove the stability of the multi-drone slung load system despite the successful simulation results due to its high-dimensional coupling characteristics \cite{maza2009multi}, \cite{cai2024robust}, and underactuated dynamics \cite{lee2017geometric}.  To address this problem, Qian and Liu \cite{qian2019path} designed a two-loop control and tracking scheme that includes an outer loop robust controller for trajectory tracking and an inner loop attitude tracker on each drone, which follows the attitude commands from the outer loop controller. Later, they proved that the overall system was Lyapunov stable \cite{qian2022robust}. They also improved the design by adding a UDE to the outer loop. Both experiments and simulations of path-following tasks with disturbances were conducted to showcase the real-world implementation capabilities. Cai \textit{et al.} \cite{cai2024robust} also used a similar hierarchical controller design and achieved Lyapunov stability, with simulations showing position convergence and attitude stabilization. Directly proving stability is also possible with multiple assumptions; Lee \cite{lee2017geometric} successfully demonstrated stability using the designed geometric controller and simplified dynamics, with simulations demonstrating the ability of this controller to stabilize with bounded tracking error. Furthermore, Gao \textit{et al.} \cite{gao2025robustness} recently proved the stability of a neuro-geometric controller for a centralized 3-drone transportation system.

\begin{figure}
    \centering{%
        \includegraphics[width=1.0\linewidth]{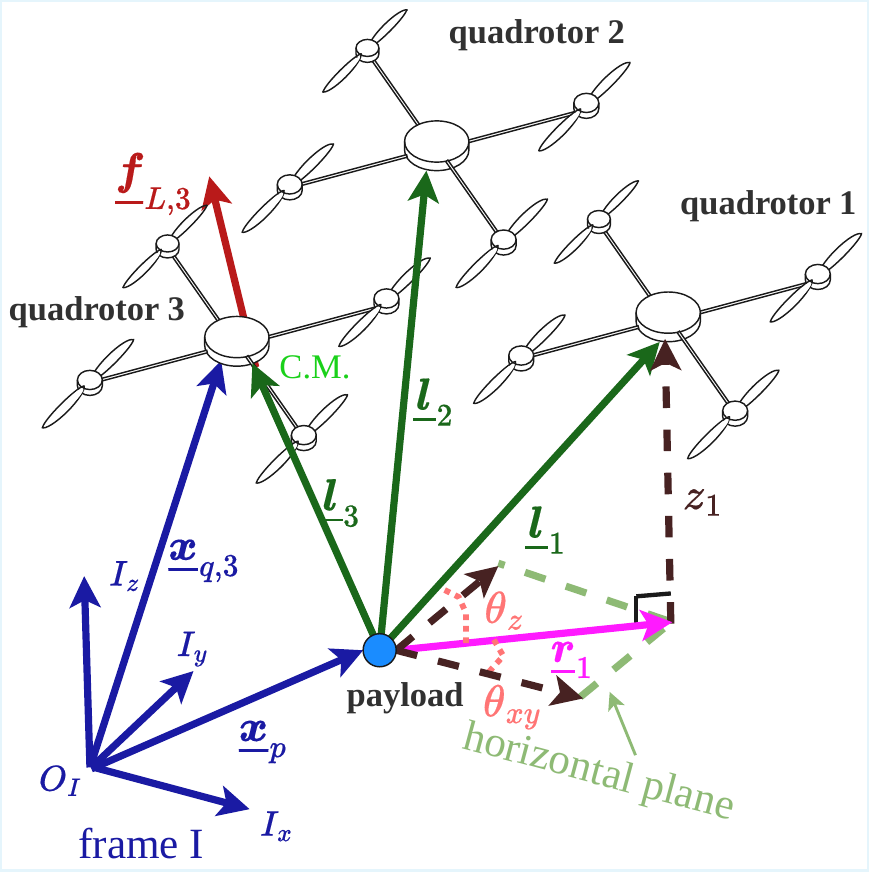}%
    }
    \caption{System geometry.}
    \label{fig: system_config}
    \vspace{-0.3cm}
\end{figure}

The complexity of the multi-drone slung payload system makes controller design challenging from a traditional control Lyapunov function (CLF) approach. Around 1998, the concept of control contraction metric (CCM) for trajectory tracking problems was proposed in \cite{lohmiller1998contraction}. Multiple studies since then have yielded a new control method using CCM on nonlinear systems \cite{manchester2017control}. The rapid development of deep learning has forged a new approach to find such a contraction metric and controller through a neural network \cite{sun2021learning}. Many advancements focus on realizing robustness has been addressed using such CCM controller design \cite{li2025neural, manchester2018robust, tsukamoto2020robust, zhao2022tube, lapandi2024meta, jin2025enhanced}. Detailed descriptions of neural CCM (N-CCM) can be found in \cite{tsukamoto2021contraction, tsukamoto2021theoretical}. However, only simplified low-dimensional cases were tested in \cite{sun2021learning}, while high-dimensional nonlinear systems may fail, such as our multi-drone payload system. On the other hand, many safety considerations were addressed in \cite{dawson2023safe}, but control saturation remains a challenge.

In this paper, we propose a robust non-linear control scheme using N-CCM for a three-drone point-mass-slung payload system. The dynamic model is derived using Kane's method. A CCM-based controller is constructed as in \cite{sun2021learning}, with a control saturation to satisfy the control constraint. The contributions and novelty of the paper are listed as follows.
\begin{enumerate}
    \item An exponentially converging controller for the multi-quadrotor slung-load system is obtained by using N-CCM. Compared with previous work \cite{qian2019path} on slung-load control, our strategy naturally inherits bounded control output to satisfy control saturation constraints while guaranteeing the stability of the system.
    \item An UDE derived from the results in \cite{qian2019path} to compensate for persistent external disturbances. We show that the UDE compensator provides a bounded and converging disturbance estimation error.
    \item The proposed controller scheme is fully modularized. By combining the classic UDE and attitude tracker adopted from \cite{qian2019path} and \cite{roza2014class} with the CCM-based baseline controller, we show that the complete closed-loop system is stable and robust.
\end{enumerate}
The rest of the paper is structured as follows. Section \ref{sec: problem_formulation} describes the dynamics and control problem. Section \ref{sec: ccm_control} states the framework of the CCM-based baseline controller. Section \ref{UDE section} and \ref{eq: attitude_tracker} provide the UDE and the attitude tracking law design. Section \ref{sec: stability proof} analyses the full-system stability. Section \ref{sec: simulation_verification} shows simulation verifications of the proposed control framework. Finally, Section \ref{sec: con} concludes the paper. 

\section{Problem Formulation}\label{sec: problem_formulation}

\subsection{Mathematical Preliminaries}
A vector is denoted as $\underline{\boldsymbol{x}}$, with $\underline{\boldsymbol{x}}_a$ as to reference $a$. Lowercase letters (i.e. $\theta$) are scalars. The identity matrix and the zero matrix are denoted as $\boldsymbol{1}$ and $\boldsymbol{0}$. Matrices are uppercase bold letters. $\boldsymbol{A} \in \mathbb{R}^{n \times m}$ denotes a $n\times m$ real matrix. The inner product of two vectors is denoted as $c = \underline{\boldsymbol{a}}^T \underline{\boldsymbol{b}}$. For $\underline{\boldsymbol{a}} \in \mathbb{R}^{n \times 1}$, $||\underline{\boldsymbol{a}}|| = \sqrt{\underline{\boldsymbol{a}}^T\underline{\boldsymbol{a}}}$. Let $\underline{\boldsymbol{\phi}} \in \mathbb{R}^{3 \times 1} = \begin{bmatrix}
\phi_1 & \phi_2 & \phi_3
\end{bmatrix}^T$ be a vector, a skew-symmetric matrix $\underline{\boldsymbol{\phi}}^{\times}$ is defined as:
\vspace{-0.1cm}
\begin{equation}\label{eq:cross_matrix}
\underline{\boldsymbol{\phi}}^{\times}:=
\begin{bmatrix}
0 & -\phi_3 & \phi_2\\
\phi_3 & 0 & -\phi_1\\
-\phi_2 & \phi_1 & 0\\
\end{bmatrix}.
\vspace{-0.1cm}
\end{equation}
Similarly, given a skew-symmetric matrix $\boldsymbol{S} = -\boldsymbol{S}^T \in \mathbb{R}^{3 \times 3}$, we denote $\boldsymbol{S}^{\vee} = \begin{bmatrix}
S_{32}& S_{13} & S_{21}
\end{bmatrix}^T$. The symmetric part of a square matrix $\boldsymbol{A}$ is denoted as $sym(\boldsymbol{A}) = \frac{\boldsymbol{A} + \boldsymbol{A}^T}{2} $. $\boldsymbol{A}_{ann}$ is the annihilator matrix of $\boldsymbol{A}$ such that $\boldsymbol{A}_{ann}^T \,\boldsymbol{A} = \boldsymbol{0}$. Matrix inequalities are denoted by curly arrows, where $\boldsymbol{A}\prec\boldsymbol{0}$ indicates that $\boldsymbol{A}$ is strictly negative definite. $diag(\boldsymbol{A}_k)$, $vstack(\boldsymbol{A}_k)$ and $hstack(\boldsymbol{A}_k)$ represents diagonal, vertical and horizontal concatenation of matrix $\boldsymbol{A}_k$ for $k=1,2,...$. The vectors $\underline{\boldsymbol{e}}_i$, for $i=1,2,3$, represent standard Euclidean basis vectors.

\subsection{System Dynamics}
According to the system geometry in Fig.\ref{fig: system_config}, a point-mass slung payload with mass $m_p$ is carried by three quadrotors with position $\underline{\boldsymbol{x}}_p$ in the inertial frame, each producing a three-dimensional (3D) lift force $\underline{\boldsymbol{f}}_{L,j}$. The mass of each quadrotor is $m_j$ with position $\underline{\boldsymbol{x}}_{q,j}$ in the inertial frame, $j = 1,2,3$. The cables are attached at the center of mass of the quadrotors such that the attitude dynamics of the quadrotors are decoupled from the payload dynamics. The cable vector defined in frame I (inertial frame) is $\underline{\boldsymbol{l}}_j \in \mathbb{R}^3$, with equal length $l = ||\underline{\boldsymbol{l}}_j||$. Each cable forms a horizontal projection $\underline{\boldsymbol{r}}_j$, the vertical and horizontal angles to this projection are $\theta_{z}$ and $\theta_{xy}$. The cable vectors can be separated into horizontal (x-y plane as $x_j$ and $y_j$ coordinates) and z-axis as follows:
\begin{equation}\label{cable vector}
\underline{\boldsymbol{l}}_j = \begin{bmatrix}
\underline{\boldsymbol{r}}_j\\
\sqrt{l^2 -\underline{\boldsymbol{r}}_j^T\underline{\boldsymbol{r}}_j}
\end{bmatrix}, \text{ } \underline{\boldsymbol{r}}_j = \begin{bmatrix}
x_j\\
y_j
\end{bmatrix},
\end{equation}

We let $z_j = \sqrt{l^2 - \underline{\boldsymbol{r}}_j^T \underline{\boldsymbol{r}}_j}$. The time derivative of the cable vector and an auxiliary matrix $\boldsymbol{B}_j$ are given below:
\begin{equation}\label{l_dot and B_j}
\dot{\underline{\boldsymbol{l}}}_j = \begin{bmatrix}

\underline{\boldsymbol{v}}_j\vspace{0.05cm}\\ 
-\frac{\underline{\boldsymbol{r}}_j^T \underline{\boldsymbol{v}}_j}{z_j}
\end{bmatrix} = \begin{bmatrix}
\boldsymbol{1}_{2\times 2}\\
-\frac{\underline{\boldsymbol{r}}_j^T}{z_j}
\end{bmatrix} \underline{\boldsymbol{v}}_j = \boldsymbol{B}_j \underline{\boldsymbol{v}}_j,
\end{equation}
with $\underline{\boldsymbol{v}}_j$ as the cable velocity in the x-y plane. It is trivial to verify the following relation:
\vspace{-0.02cm}
\begin{equation}\label{eq: B_j_column}
\vspace{-0.02cm}
\boldsymbol{B}_j^T \underline{\boldsymbol{l}}_j = \boldsymbol{0}.   
\end{equation}

Hence, the columns of $\boldsymbol{B}_j$ are perpendicular to the vector $\underline{\boldsymbol{l}}_j$. The detailed derivation of our system dynamics (i.e., the inertial matrix $\boldsymbol{M}$, the gyroscopic matrix $\boldsymbol{C}$, payload gravitational force $\underline{\boldsymbol{f}}_{g,p}$, control matrix $\boldsymbol{H}$, and disturbance matrix $\boldsymbol{H}_{\delta}$) can be found in Sec. 1 of the support document\footnote{Support document at \url{https://github.com/maxl-xy/ACC2026}.\label{git}} using Kane's method. We can compensate for the quadrotor's weight by setting $\underline{\boldsymbol{f}}_{L,j} = -m_j \underline{\boldsymbol{g}}_I + \delta \underline{\boldsymbol{f}}_{L,j}$,
such that the control signal $\delta \underline{\boldsymbol{f}}_{L,j}$ already counters the gravity on the quadrotors. The total payload system with velocity vector $\underline{\boldsymbol{u}} = \begin{bmatrix}
\underline{\boldsymbol{v}}_p^T & \underline{\boldsymbol{v}}_1^T & \underline{\boldsymbol{v}}_2^T &
\underline{\boldsymbol{v}}_3^T
\end{bmatrix}^T$ and the full state $\underline{\boldsymbol{x}} = \begin{bmatrix}
\underline{\boldsymbol{x}}_p^T & \underline{\boldsymbol{r}}_1^T & \underline{\boldsymbol{r}}_2^T & \underline{\boldsymbol{r}}_3^T & 
\underline{\boldsymbol{u}}^T
\end{bmatrix}^T$ is defined as follows:
\begin{equation}\label{eq: system_model}
\Sigma_p : \left\{\begin{tabular}{ l c r }
  $ \boldsymbol{M}\underline{\boldsymbol{\dot{u}}}+\boldsymbol{C}\underline{\boldsymbol{u}}= \underline{\boldsymbol{f}}_{g,p} + \boldsymbol{H} \underline{\boldsymbol{\zeta}} + \boldsymbol{H}_{\delta}\underline{\boldsymbol{\delta}}$ \\
  $\underline{\boldsymbol{\dot{x}}}_p = \underline{\boldsymbol{v}}_p$ \\
  $\underline{\boldsymbol{\dot{r}}}_j =  \underline{\boldsymbol{v}}_j$
\end{tabular}\right.
\end{equation}
where $\underline{\boldsymbol{v}}_p$ is the payload velocity, $\underline{\boldsymbol{\zeta}}$ is the control input, and $\underline{\boldsymbol{\delta}}$ is the disturbance vector. Given that the total mass of the system is $m_t = m_p + m_1 + m_2 + m_3$, the system matrices are the following:
\begin{equation}\label{eq: system matrices}
\begin{aligned}
& \boldsymbol{M} = \begin{bmatrix}
m_{t} \boldsymbol{1}_{3\times 3} & hstack(m_j \boldsymbol{B}_j) \\
vstack(m_j \boldsymbol{B}_j^T) & diag(m_j \boldsymbol{B}_j^T \boldsymbol{B}_j)
\end{bmatrix},\\
& \boldsymbol{C} = \begin{bmatrix}
\boldsymbol{0}_{3\times 3} & hstack(m_j \dot{\boldsymbol{B}}_j) \\
\boldsymbol{0}_{6\times 3} & diag(m_j \boldsymbol{B}_j^T \dot{\boldsymbol{B}}_j)
\end{bmatrix}, \quad
\underline{\boldsymbol{f}}_{g,p} = \begin{bmatrix}
m_p \underline{\boldsymbol{g}}_I\\
\boldsymbol{0}_{6\times 1}
\end{bmatrix},\\
& \boldsymbol{H} = \begin{bmatrix}
hstack(\boldsymbol{1}_{3\times 3})\\
diag(\boldsymbol{B}_j^T)
\end{bmatrix},\quad \underline{\boldsymbol{\zeta}} = \begin{bmatrix}
vstack(\delta \underline{\boldsymbol{f}}_{L,j})
\end{bmatrix}, \\
& \boldsymbol{H}_{\delta} = \begin{bmatrix}
\boldsymbol{1}_{3\times 3} & hstack(\boldsymbol{1}_{3\times 3})\\
\boldsymbol{0}_{6\times 3} & diag(\boldsymbol{B}_j^T)
\end{bmatrix},\quad \underline{\boldsymbol{\delta}} = \begin{bmatrix}
\underline{\boldsymbol{\delta}}_p\\
vstack(\underline{\boldsymbol{\delta}}_j)
\end{bmatrix}. 
\end{aligned}
\end{equation}
After this manipulation, we calculated that:
\begin{equation}
\begin{aligned}
& \boldsymbol{G}(\underline{\boldsymbol{x}}) = \begin{bmatrix}
\boldsymbol{0}_{9\times 9}\\
\boldsymbol{M}^{-1} \boldsymbol{H}
\end{bmatrix},\quad \boldsymbol{G}_{\delta}(\underline{\boldsymbol{x}}) = \begin{bmatrix}
\boldsymbol{0}_{9\times 12}\\
\boldsymbol{M}^{-1} \boldsymbol{H}_{\delta}
\end{bmatrix},\\
& \underline{\boldsymbol{f}}(\underline{\boldsymbol{x}}) = \begin{bmatrix}
\underline{\boldsymbol{u}}\\
\boldsymbol{M}^{-1} \left( \underline{\boldsymbol{f}}_{g,p} - \boldsymbol{C} \underline{\boldsymbol{u}} \right)
\end{bmatrix}.
\end{aligned}
\end{equation}
The final control-affine system is given below:
\begin{equation}\label{eq: ctrl affine_noise}
\dot{\underline{\boldsymbol{x}}} = \underline{\boldsymbol{f}}(\underline{\boldsymbol{x}}) + \boldsymbol{G}(\underline{\boldsymbol{x}}) \underline{\boldsymbol{\zeta}} + \boldsymbol{G}_{\delta}(\underline{\boldsymbol{x}}) \underline{\boldsymbol{\delta}}
\end{equation}

% briefly state the mission of the controller
The goal of this paper is to design a feedback controller $\underline{\boldsymbol{\zeta}}$ such that the states $\underline{\boldsymbol{x}}$ initialized in the neighborhood of $\underline{\boldsymbol{x}}^*$ converges to the reference under external disturbance $\underline{\boldsymbol{\delta}}$ while $\underline{\boldsymbol{\zeta}}$ is bounded by some control saturation constraints.

\section{Neural Robust Control with CCM} \label{sec: ccm_control}
The CCM-based control law is realized by training neural networks representing $ \underline{\boldsymbol{\zeta}}_{nn}$ and $\boldsymbol{P}(\underline{\boldsymbol{x}}, t)$ to satisfy differential stability conditions simultaneously. We adopt the framework presented in \cite{manchester2017control} for training. During training, we assume zero external disturbance, i.e. $\underline{\boldsymbol{\delta}} = \underline{\boldsymbol{0}}$. The external disturbance is compensated later by a UDE introduced in Section \ref{UDE section}. Hence, the control-affine model for training is:
\begin{equation}\label{eq: ctrl_affine}
\underline{\dot{\boldsymbol{x}}} = \underline{\boldsymbol{f}}(\underline{\boldsymbol{x}}) + \boldsymbol{G}(\underline{\boldsymbol{x}}) \underline{\boldsymbol{\zeta}}_{nn}, 
\end{equation}
where $\underline{\boldsymbol{x}} \in \mathbb{R}^n$ is the state and $\underline{\boldsymbol{\zeta}}_{nn} \in \mathbb{R}^m$ is the control input. A smooth control law can be found as
\begin{equation}\label{neural control signal}
\underline{\boldsymbol{\zeta}}_{nn} = \underline{\boldsymbol{k}}(\underline{\boldsymbol{x}}, \underline{\boldsymbol{x}}^*, \underline{\boldsymbol{k}}^*; \theta_{K_1}, \theta_{K_2}),
\end{equation}
such that $\underline{\boldsymbol{x}}^*$ and $\underline{\boldsymbol{k}}^*$ are the bounded desired state and control signal. $\theta_{K_1}$ and $\theta_{K_2}$ are learned parameters from two fully connected neural networks $\boldsymbol{K}_1$ and $\boldsymbol{K}_2$.  We choose $\underline{\boldsymbol{k}} = \boldsymbol{K}_2 \text{ }tanh(\boldsymbol{K}_1 (\underline{\boldsymbol{x}} - \underline{\boldsymbol{x}}^*)) + \underline{\boldsymbol{k}}^*$, where $tanh(\cdot)$ is the element-wise hyperbolic tangent function, such that when $\underline{\boldsymbol{x}} \rightarrow \underline{\boldsymbol{x}}^*$, $\underline{\boldsymbol{k}} \rightarrow \underline{\boldsymbol{k}}^*$. 
We also choose $\boldsymbol{P} = \boldsymbol{W}^{-1}$ where $\boldsymbol{W}$ is a dual metric of the CCM defined as $\boldsymbol{W} = \boldsymbol{L}(\underline{\boldsymbol{x}}, \theta_{L})^T \boldsymbol{L}(\underline{\boldsymbol{x}}, \theta_L) + \underline{w}\boldsymbol{1}$. $\theta_L$ are learned parameters from neural network $\boldsymbol{L}$, and $\underline{w}$ is a positive constant that represents the smallest eigenvalue of the dual metric. Note that the CCM is only a function of $\underline{\boldsymbol{x}}$ as the system dynamics are time-independent. Such an approach was used in \cite{sun2021learning} to prove the global stability, and the trajectories contract exponentially with rate $\lambda>0$ if the following contraction conditions are satisfied:
\begin{equation}\label{contraction condition}
\dot{\boldsymbol{P}} + sym \Big(\boldsymbol{P}(\boldsymbol{A} + \boldsymbol{G}(\underline{\boldsymbol{x}})  \boldsymbol{K}) \Big) + 2\lambda \boldsymbol{P} \prec \boldsymbol{0},
\end{equation}
\begin{equation}\label{contraction condition bound}
    \boldsymbol{W}-\overline{w}\cdot \boldsymbol{1} \prec \boldsymbol{0},
\end{equation}
where $\boldsymbol{A} = \frac{\partial \underline{\boldsymbol{f}}}{\partial \underline{\boldsymbol{x}}} + \sum_{i=1}^m \frac{\partial \underline{\boldsymbol{g}}_i}{\partial \underline{\boldsymbol{x}}} \zeta_{nn,i}$, $\boldsymbol{K} = \frac{\partial \underline{\boldsymbol{k}}}{\partial \underline{\boldsymbol{x}}}$, and $\underline{\boldsymbol{g}}_i$ is the $i^{th}$ column of $\boldsymbol{G}$, $\zeta_{nn,i}$ is the $i^{th}$ element of $\underline{\boldsymbol{\zeta}}_{nn}$. $\overline{w}$ is a positive constant that represents the largest eigenvalue of $\boldsymbol{W}$. The authors in \cite{sun2021learning} also incorporated the dual conditions.
\begin{equation}\label{dual condition 1}
\boldsymbol{G}_{ann}^T \left( -\frac{\partial \boldsymbol{W}}{\partial \underline{\boldsymbol{x}}}\underline{\boldsymbol{f}} + sym \Big(\frac{\partial \underline{\boldsymbol{f}}}{\partial \underline{\boldsymbol{x}}} \boldsymbol{W} \Big) + 2\lambda \boldsymbol{W} \right) \boldsymbol{G}_{ann} \prec \boldsymbol{0},
\end{equation}
\begin{equation}\label{dual condition 2}
\boldsymbol{G}_{ann}^T \left(  \frac{\partial \boldsymbol{W}}{\partial \underline{\boldsymbol{x}}}\underline{\boldsymbol{g}}_i - sym \Big(\frac{\partial \underline{\boldsymbol{g}}_i}{\partial \underline{\boldsymbol{x}}} \boldsymbol{W} \Big) \right) \boldsymbol{G}_{ann} = \boldsymbol{0}, i = 1,...,m.
\end{equation}
The dual metric $\boldsymbol{W}$ and the controller are trained separately using fully connected neural networks, with conditions \eqref{contraction condition}, \eqref{contraction condition bound}, \eqref{dual condition 1}, and \eqref{dual condition 2} as loss terms. To add control constraints to the neural controller, we use a saturation function $tanh(\cdot)$ at the end of the neural calculation of $\underline{\boldsymbol{k}}$ in \eqref{neural control signal}, with a saturation factor $a$, and a control bound $f_b$ to tune the domain and range of the control signal from the neural network. The reference control signal $\underline{\boldsymbol{k}}^*$ is outside of the saturation function to guarantee the desired state and control. The output control signal after the hard control constraints is
\begin{equation}\label{saturated control signal feedback}
\underline{\boldsymbol{\zeta}}_{nn,sat} = tanh(a \cdot \underline{\boldsymbol{k}}(\underline{\boldsymbol{x}}, \underline{\boldsymbol{x}}^*, \underline{\boldsymbol{0}}; \theta_{K_1}, \theta_{K_2})) \cdot f_b + \underline{\boldsymbol{k}}^*.
\end{equation}
This ensures the smoothness of $\underline{\boldsymbol{k}}$ even after saturation, while guaranteeing the desired control signal $\underline{\boldsymbol{k}}^*$.

\section{The Uncertainty and Disturbance Estimator}\label{UDE section}
\subsection{Effective Disturbances}

We decompose the disturbances on each quadrotor into two components: $\underline{\boldsymbol{\delta}}_{\bot,j}$ and $\underline{\boldsymbol{\delta}}_{\parallel,j}$ which are the components of $\underline{\boldsymbol{\delta}}_j$ that are perpendicular and parallel to $\underline{\boldsymbol{l}}_j$, respectively. $\underline{\boldsymbol{\delta}}_{T}$ is defined as the effective disturbance on the payload. These disturbances are obtained in the following way:
\begin{equation} \label{eq: delta_bot_para}
\begin{aligned}
& \left\{\begin{tabular}{ l c r }
  $ \underline{\boldsymbol{\delta}}_{\parallel,j} = \underline{\boldsymbol{l}}_j \underline{\boldsymbol{l}}_j^T\underline{\boldsymbol{\delta}}_j/l^2 $ \\
  $ \underline{\boldsymbol{\delta}}_{\bot,j} = \underline{\boldsymbol{\delta}}_j - \underline{\boldsymbol{\delta}}_{\parallel,j}$ \\
\end{tabular}\right.;\quad \underline{\boldsymbol{\delta}}_{T} = \underline{\boldsymbol{\delta}}_p + \sum^{N}_{j=1} \underline{\boldsymbol{\delta}}_{\parallel,j}
\end{aligned}
\end{equation}

The estimated values of $\underline{\boldsymbol{\delta}}_{j}$ and $\underline{\boldsymbol{\delta}}_{T}$ are $\underline{\boldsymbol{\hat{\delta}}}_{j}$ and $\underline{\boldsymbol{\hat{\delta}}}_{T}$, respectively. The estimation errors are $\underline{\boldsymbol{\tilde{\delta}}}_{j} = \underline{\boldsymbol{\hat{\delta}}}_{j} - \underline{\boldsymbol{\delta}}_{j}$ and $\underline{\boldsymbol{\tilde{\delta}}}_{T} = \underline{\boldsymbol{\hat{\delta}}}_{T} - \underline{\boldsymbol{\delta}}_{T}$.

\begin{assumption}\label{assm: disturbance_assm}
All disturbances are bounded. $\boldsymbol{\dot{\delta}}_T \approx \boldsymbol{0}$ and $\boldsymbol{\dot{\delta}}_j \approx \boldsymbol{0}$ are assumed as reasonable engineering treatments near hover in near-calm winds for a typical robust control design \cite{qian2022robust}. The following identities are used in the subsequent stability analysis:

\begin{equation}\label{eq: effective_disturbances_identity}
    \underline{\boldsymbol{\delta}}_p + \sum^{N}_{j=1}  \underline{\boldsymbol{\delta}}_j = \underline{\boldsymbol{\delta}}_{T} + \sum^{N}_{j=1} \underline{\boldsymbol{\delta}}_{\bot,j}.
\end{equation}
\end{assumption}
\subsection{The Disturbance Estimation Law}
The UDE technique in Ref. \cite{qian2022robust} is used to derive the disturbance estimation law. We examine the cable swing dynamics in $\Sigma_p$ in (\ref{eq: system_model}) and (\ref{eq: system matrices}), resulting in the following dynamics for cable acceleration (see Sec. 2 of support document\footnotemark[1] for details):
\begin{equation}\label{eq: quadrotor_relative_motion}
\begin{aligned}
& m_j\boldsymbol{B}_j^T(\underline{\boldsymbol{\dot{v}}}_p  + \boldsymbol{B}_j\underline{\boldsymbol{\dot{v}}}_j + \dot{\boldsymbol{B}}_j\underline{\boldsymbol{v}}_j ) = m_j\boldsymbol{B}_j^T\frac{d \underline{\boldsymbol{v}}_{q, j}}{dt} \\
& = \boldsymbol{B}_j^T(\delta\underline{ \boldsymbol{f}}_{L,j}+ \underline{\boldsymbol{\delta}}_j) = \boldsymbol{B}_j^T(\delta\underline{ \boldsymbol{f}}_{L,j} + \underline{\boldsymbol{\delta}}_{\bot,j}).
\end{aligned}
\end{equation}
The inertial velocity of each quadrotor is $\underline{\boldsymbol{v}}_{q,j} = \underline{\boldsymbol{v}}_p + \boldsymbol{B}_j\underline{\boldsymbol{v}}_j$.  According to (\ref{eq: B_j_column}) and (\ref{eq: delta_bot_para}), we know that $\boldsymbol{B}_j^T\underline{\boldsymbol{\delta}}_{\parallel,j}=\boldsymbol{0}$. Similarly, the estimation value and error of $\underline{\boldsymbol{\delta}}_{\bot,j}$ have the following property: 
\begin{equation}\label{eq: delta_bot_estimation}
\begin{aligned}
& \underline{\boldsymbol{\hat{\delta}}}_{\bot,j} = (\boldsymbol{1}-\underline{\boldsymbol{l}}_j \underline{\boldsymbol{l}}_j^T/l^2) \underline{\boldsymbol{\hat{\delta}}}_j, \quad \underline{\boldsymbol{\tilde{\delta}}}_{\bot,j} = (\boldsymbol{1} - \underline{\boldsymbol{l}}_j \underline{\boldsymbol{l}}_j^T/l^2) \underline{\boldsymbol{\tilde{\delta}}}_j.
\end{aligned}
\end{equation}
$\mathfrak{B}_j = \boldsymbol{B}_j (\boldsymbol{B}_j^T\boldsymbol{B}_j)^{-1}\boldsymbol{B}_j^T$ are a series of auxiliary matrices. The dynamics of the estimator for $\underline{\boldsymbol{\tilde{\delta}}}_{j}$  is set to:
\begin{equation}\label{eq: delta_j_err_dynamics}
\begin{aligned}
\underline{\boldsymbol{\dot{\hat{\delta}}}}_j = \underline{\boldsymbol{\dot{\tilde{\delta}}}}_j = - \kappa_j\mathfrak{B}_j\underline{\boldsymbol{\tilde{\delta}}}_{\bot,j}.
\end{aligned}
\end{equation}
$\kappa_j$ is a positive rate constant. Note that based on the design procedure in \cite{qian2022robust} and Assumption \ref{assm: disturbance_assm}, $\underline{\boldsymbol{\dot{\delta}}}_j \approx \boldsymbol{0}$. Hence, the differential form of the estimated disturbance $\underline{\boldsymbol{\hat{\delta}}}_j$ is:
\begin{equation}\label{eq: hat_delta_j_diff}
\begin{aligned}
& \underline{\boldsymbol{\dot{\hat{\delta}}}}_j = - \kappa_j\mathfrak{B}_j(\underline{\boldsymbol{\hat{\delta}}}_j - \underline{\boldsymbol{\delta}}_j) = \kappa_j\mathfrak{B}_j (m_j\underline{\boldsymbol{\dot{v}}}_{q,j} - \delta \underline{\boldsymbol{f}}_{L,j} - \underline{\boldsymbol{\hat{\delta}}}_j ).\\
\end{aligned}
\end{equation}
The final update law in integral form of $\underline{\boldsymbol{\hat{\delta}}}_{\bot,j}$ is:
\begin{equation}\label{eq: hat_delta_j}
\begin{aligned}
& \underline{\boldsymbol{\hat{\delta}}}_j = \int^{t}_0 \kappa_j \mathfrak{B}_j (m_j\underline{\boldsymbol{\dot{v}}}_{q,j} - \delta \underline{\boldsymbol{f}}_{L,j} - \underline{\boldsymbol{\hat{\delta}}}_j)d\tau,
\end{aligned}
\end{equation}
where $\boldsymbol{\dot{v}}_{q,j}$ is the acceleration of each quadrotor measured by the onboard IMU. It can be calculated using the quadrotor's attitude and the raw acceleration feedback. Here $\delta \underline{\boldsymbol{f}}_{L,j}$ is the actual lift calculated based on the thrust model from system identification and quadrotor attitude. After obtaining $\underline{\boldsymbol{\hat{\delta}}}_{\bot,j}$, we set the error dynamics of $\underline{\boldsymbol{\tilde{\delta}}}_T$ as follows:
\begin{equation}\label{eq: dynamics_delta_T}
\begin{aligned}
& \underline{\boldsymbol{\dot{\tilde{\delta}}}}_T/\lambda_{T} = - \underline{\boldsymbol{\tilde{\delta}}}_T - \sum^{N}_{j=1} \underline{\boldsymbol{\tilde{\delta}}}_{\bot,j} \Rightarrow    - \underline{\boldsymbol{\tilde{\delta}}}_T  = \underline{\boldsymbol{\dot{\tilde{\delta}}}}_T/\lambda_{T} + \sum^{N}_{j=1} \underline{\boldsymbol{\tilde{\delta}}}_{\bot,j}.
\end{aligned}
\end{equation}
$\lambda_T$ is a positive rate constant. For our system, $N=3$ According to Assumption \ref{assm: disturbance_assm}, $ \underline{\boldsymbol{\dot{\delta}}}_T \approx 0$ and $ \underline{\boldsymbol{\delta}}_T = \underline{\boldsymbol{\hat{\delta}}}_T  - \underline{\boldsymbol{\tilde{\delta}}}_T $. Hence $\underline{\boldsymbol{\dot{\tilde{\delta}}}}_T$ has the following relationship:
\begin{equation}\label{eq: multi_force_estimation_dynamics}
\begin{aligned}
&\underline{\boldsymbol{\dot{\tilde{\delta}}}}_T/\lambda_{T}  = (\underline{\boldsymbol{\dot{\hat{\delta}}}}_T - \underline{\boldsymbol{\dot{\delta}}}_T )/\lambda_{T}  = \underline{\boldsymbol{\dot{\hat{\delta}}}}_T/\lambda_{T} \\
\end{aligned}
\end{equation}
We can extract the payload translation dynamics from (\ref{eq: system_model}) and (\ref{eq: system matrices}) as follows (see Sec. 2 of support document\footnotemark[1] for details):
\begin{equation}\label{eq: payload_translation_dynamics}
\begin{aligned}    
& \frac{d}{dt}\big(m_t\underline{\boldsymbol{v}}_p + \sum^{N}_{j=1} m_j \boldsymbol{B}_j \underline{\boldsymbol{v}}_j\big) = \underline{\boldsymbol{\delta}}_T + m_p \underline{\boldsymbol{g}}_I + \sum_{j=1}^{N} (\delta \underline{\boldsymbol{f}}_{L,j} +  \underline{\boldsymbol{\delta}}_{\bot,j}). 
\end{aligned}
\end{equation}
By inserting (\ref{eq: dynamics_delta_T}) and (\ref{eq: multi_force_estimation_dynamics}) into (\ref{eq: payload_translation_dynamics}) and applying $\underline{\boldsymbol{\delta}}_T = \underline{\boldsymbol{\hat{\delta}}}_T  - \underline{\boldsymbol{\tilde{\delta}}}_T$, we have the following update law:
\begin{equation}\label{eq: hat_delta_T_diff}
\begin{aligned}    
& d(m_t\underline{\boldsymbol{v}}_p + \sum^{N}_{j=1} m_j \boldsymbol{B}_j \underline{\boldsymbol{v}}_j)/dt \\
& = \underline{\boldsymbol{\hat{\delta}}}_T  +  \underline{\boldsymbol{\dot{\hat{\delta}}}}_T/\lambda_{T} + m_p \underline{\boldsymbol{g}}_I + \sum_{j=1}^{N} (\delta \underline{\boldsymbol{f}}_{L,j} + \underline{\boldsymbol{\hat{\delta}}}_{\bot,j})\\
& \Rightarrow \underline{\boldsymbol{\dot{\hat{\delta}}}}_T/\lambda_{T} =  d(m_t\underline{\boldsymbol{v}}_p  + \sum^{N}_{j=1}m_j \boldsymbol{B}_j \underline{\boldsymbol{v}}_j)/dt - \underline{\boldsymbol{\hat{\delta}}}_T \\
& - m_p\underline{\boldsymbol{g}}_I - \sum_{j=1}^{N} (\delta \underline{\boldsymbol{f}}_{L,j} + \boldsymbol{\hat{\delta}}_{\bot,j})
\end{aligned}
\end{equation}
It is trivial to verify that the integral form of (\ref{eq: hat_delta_T_diff}) is equivalent to (\ref{eq: detla_T_hat_int_form}). We do not have a measurement of $\boldsymbol{\dot{v}}_p$ because we assume that no IMU is installed on the payload; therefore, the integral form of the above utilizes only velocity feedback to construct the estimation. The final expression of $\boldsymbol{\hat{\delta}}_T$ becomes: 
\begin{equation}\label{eq: detla_T_hat_int_form}
\begin{aligned}
& \underline{\boldsymbol{\hat{\delta}}}_{T} =  \lambda_{T} \Big[m_t\underline{\boldsymbol{v}}_p + \sum^{N}_{j=1} m_j \boldsymbol{B}_j \underline{\boldsymbol{v}}_j\\
& - \int^{t}_0 \sum^{N}_{j=1}( \delta \underline{\boldsymbol{f}}_{L,j} + \underline{\boldsymbol{\hat{\delta}}}_{\bot,j}) + \underline{\boldsymbol{\hat{\delta}}}_{T} + m_p \underline{\boldsymbol{g}}_I d \tau\Big].
\end{aligned}
\end{equation}
Once (\ref{eq: hat_delta_j}) and (\ref{eq: detla_T_hat_int_form}) are obtained, the control force $\underline{\boldsymbol{f}}_{\delta}$ balancing the estimated disturbances can be obtained as:
\begin{equation}\label{eq: delta_T_allocation}
\begin{aligned}
& \underline{\boldsymbol{f}}_{\delta, j} = - \underline{\boldsymbol{n}}_j\hat{\delta}_{T, j} -   \underline{\boldsymbol{\hat{\delta}}}_{\bot,j}, \quad \underline{\boldsymbol{f}}_{\delta} = vstack(\underline{\boldsymbol{f}}_{\delta, j})\\
& \begin{bmatrix}
\hat{\delta}_{T, 1} &
\hat{\delta}_{T, 2} &
\hat{\delta}_{T, 3}
\end{bmatrix}^T = \begin{bmatrix}
     \underline{\boldsymbol{n}}_1, \underline{\boldsymbol{n}}_2, \underline{\boldsymbol{n}}_3
\end{bmatrix}^{-1}\underline{\boldsymbol{\hat{\delta}}}_{T}
\end{aligned}
\end{equation}
where $\underline{\boldsymbol{n}}_j=\underline{\boldsymbol{l}}_j/l$.
Since the cable vectors $\underline{\boldsymbol{l}}_j$ point to different directions generated by the trajectory planner, the linear equation in (\ref{eq: delta_T_allocation}) is guaranteed to provide a unique solution.

\section{The Quadrotor Attitude Control Law}\label{eq: attitude_tracker}
Once $\underline{\boldsymbol{\zeta}}_{nn,sat}$ and $ \underline{\boldsymbol{f}}_{\delta}$ are obtained, we can calculate the total desired control force $\underline{\boldsymbol{\zeta}}_c = \underline{\boldsymbol{\zeta}}_{nn,sat} + \underline{\boldsymbol{f}}_{\delta}$. From Fig.\ref{Complete closed-loop system}, the total desired force for the $j^{th}$ drone is $\underline{\boldsymbol{f}}_{Lc, j}$, we adopt a classic attitude tracker in \cite{roza2014class} to achieve $\underline{\boldsymbol{f}}_{Lc, j}$. The total lift from the propellers is $f_j = ||\underline{\boldsymbol{f}}_{Lc,j}||$. A command yaw angle $\psi$ is picked for each quadrotor. The lift is assumed along the z-axis of the quadrotor, i.e. $\underline{\boldsymbol{n}}_z = \underline{\boldsymbol{f}}_{Lc,j}/f_j$. The reference attitude of the drone $\boldsymbol{R}_{Ij,d}$ is obtained in the following way:
\begin{equation}\label{eq: attitude_extraction}
\begin{aligned}
& \underline{\boldsymbol{\tilde{n}}}_{x} = [
\cos\psi \quad \sin \psi \quad  -(\cos \psi n_{z,x} + \sin \psi n_{z,y})/n_{z,z}]^T; \\
& \underline{\boldsymbol{n}}_{x} = \underline{\boldsymbol{\tilde{n}}}_{x}/||\underline{\boldsymbol{\tilde{n}}}_{x}||; \quad \underline{\boldsymbol{n}}_{y} = \underline{\boldsymbol{n}}_{z}^{\times}\underline{\boldsymbol{n}}_{x}/||\underline{\boldsymbol{n}}_{z}^{\times}\underline{\boldsymbol{n}}_{x}||; \\ & \boldsymbol{R}_{Ij,d} = [
\underline{\boldsymbol{n}}_{x} \quad \underline{\boldsymbol{n}}_{y} \quad \underline{\boldsymbol{n}}_{z}],
\end{aligned}
\end{equation}
where $\underline{\boldsymbol{n}}_{z,x}$ and $\underline{\boldsymbol{n}}_{z,y}$ are the $x$ and $y$ components of $\underline{\boldsymbol{n}}_z$ respectively. We cite Section VI.C of Ref. \cite{roza2014class} to obtain an almost global asymptotically stable (AGAS) attitude tracker. First, define $\underline{\boldsymbol{\omega}}_{d,j}$ as the desired angular velocity, and $\mathcal{\tilde{X}}_{rot,j} = \{\underline{\boldsymbol{\tilde{\omega}}}_j,\tilde{\boldsymbol{R}}_j\}$ as the attitude tracking error of the $j^{th}$ drone. 
Once $\boldsymbol{R}_{Ij,d}$, $\underline{\boldsymbol{\omega}}_{d,j}$, and $\underline{\boldsymbol{\dot{\omega}}}_{d,j}$ are calculated based on $\underline{\boldsymbol{f}}_{Lc,j}$, the following attitude control law is used:
\begin{equation}\label{eq: attitude_control}
\begin{aligned}
&\underline{\boldsymbol{\tau}}_j = -b_{\omega} \underline{\boldsymbol{\tilde{\omega}}}_j - b_r \underline{\boldsymbol{e}}_{r,j} - \underline{\boldsymbol{\tilde{\omega}}}_j^{\times}\boldsymbol{J}\underline{\boldsymbol{\tilde{\omega}}}_j+ \underline{\boldsymbol{\omega}}_j^{\times}\boldsymbol{J}\underline{\boldsymbol{\omega}}_j \\
& - \boldsymbol{J}(\underline{\boldsymbol{\tilde{\omega}}}_j^{\times}\tilde{\boldsymbol{R}}_j^T \underline{\boldsymbol{\omega}}_{d,j}- \tilde{\boldsymbol{R}}_j^T\underline{\boldsymbol{\dot{\omega}}}_{d,j})  \\
\end{aligned}
\end{equation}
where $\tilde{\boldsymbol{R}}_j = \boldsymbol{R}_{Ij,d}^T\boldsymbol{R}_{Ij}$, $\underline{\boldsymbol{\omega}}_{d,j} = (\boldsymbol{R}_{Ij,d}^T\dot{\boldsymbol{R}}_{Ij,d})^{\vee}$, $\underline{\boldsymbol{\tilde{\omega}}}_j = \underline{\boldsymbol{\omega}}_j-\boldsymbol{R}_j^T\underline{\boldsymbol{\omega}}_{d,j} $, and $\underline{\boldsymbol{e}}_{r,j} =\sum_{i=1}^{3}\underline{\boldsymbol{e}}_i^{\times}\tilde{\boldsymbol{R}}_j\underline{\boldsymbol{e}}_i$. $b_{\omega}$ and $b_r$ are positive control gains and $\boldsymbol{J}$ is the moment of inertia of the drones (see Sec. 3 of support document\footnotemark[1] for details). According to Ref. \cite{roza2014class}, we conclude that with the AGAS attitude tracker in (\ref{eq: attitude_control}), $ \underline{\boldsymbol{\zeta}}_e = \underline{\boldsymbol{\zeta}} - \underline{\boldsymbol{\zeta}}_c \rightarrow 0$ as $t \rightarrow \infty$, where $\underline{\boldsymbol{\zeta}} = vstack(\boldsymbol{R}_{Ij} \underline{\boldsymbol{e}}_3 f_j + m_j \underline{\boldsymbol{g}}_I)$.

\begin{figure}
    \centering    \includegraphics[width=1.0\linewidth]{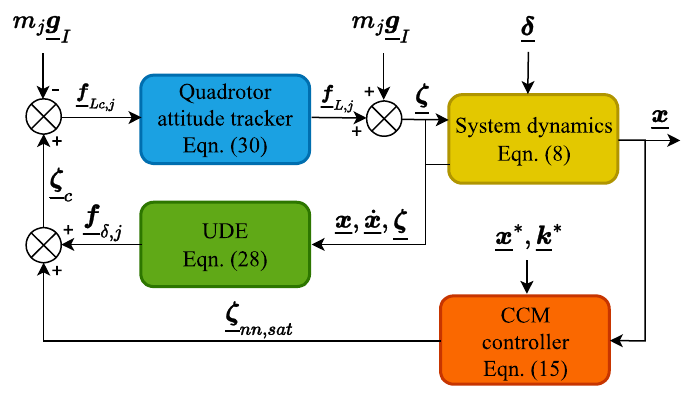}
    \caption{The complete closed-loop system.}
    \label{Complete closed-loop system}
    \vspace{-0.3cm}
\end{figure}

\section{Stability Analysis}\label{sec: stability proof}
First, we cite two important robustness results, stated as:
\begin{theorem}\label{thm: contraction_robustness}
\textbf{Theorem 2.4 of Ref. \cite{tsukamoto2021contraction}}: If the system in (\ref{eq: ctrl_affine}) is contracting, then the path integral $V_\mathcal{L}(\underline{\boldsymbol{q}}, \delta \underline{\boldsymbol{q}}, t) = \int^{\underline{\boldsymbol{\eta}}_1}_{\underline{\boldsymbol{\eta}}_0}||\boldsymbol{\Theta}(\underline{\boldsymbol{q}}, t) \delta \underline{\boldsymbol{q}}||$ of \textbf{(22) of Ref. \cite{tsukamoto2021contraction}}, where ${\underline{\boldsymbol{\eta}}}_0$ is a solution of (\ref{eq: ctrl_affine}) and ${\underline{\boldsymbol{\eta}}}_1$ is a solution of the perturbed system in \textbf{(24) of Ref. \cite{tsukamoto2021contraction}}, and $\underline{\boldsymbol{q}}$ is the virtual state of \textbf{(25) of Ref. \cite{tsukamoto2021contraction}}, exponentially converges to a bounded error ball as long as $\boldsymbol{\Theta}\underline{\boldsymbol{d}} \in \mathcal{L}_{\infty}$. Specifically, if $\exists \,\underline{m},\overline{m} \in \mathbb{R}_{>0}$ and $\exists \, \overline{d} \in \mathbb{R}_{\geq0}$ s.t. $\overline{d}=\sup_{\underline{\boldsymbol{x}},t}||\underline{\boldsymbol{d}}(\underline{\boldsymbol{x}},t)||$ and 
\begin{equation}
    \boldsymbol{1}/\overline{w} = \underline{m} \boldsymbol{1} \preceq \boldsymbol{P} \preceq \overline{m} \boldsymbol{1} = \boldsymbol{1}/\underline{w}
\end{equation}
then we have the following relation:
\begin{equation}
    ||\underline{\boldsymbol{\eta}}_1 - \underline{\boldsymbol{\eta}}_0|| \leq \sqrt{\overline{w}} \,V_{\mathcal{L}(0)}e^{-\lambda t} + \frac{\overline{d}}{\lambda} \sqrt{\frac{\overline{w}}{\underline{w}}}(1-e^{-\lambda t})
\end{equation}
\end{theorem}

\begin{lemma}\label{lem: B_p}
\textbf{Lemma 1 iv) of Ref. \cite{qian2022robust}}:  The following properties are true:
$\forall \underline{\boldsymbol{x}} \in \mathbb{R}^{3\times1} \neq 0$, we define $\underline{\boldsymbol{x}}_{\bot}$ and $\underline{\boldsymbol{x}}_{\parallel}$ as its components perpendicular and parallel to $\underline{\boldsymbol{l}}_j$. Then $\underline{\boldsymbol{x}}^T\mathfrak{B}_j\underline{\boldsymbol{x}} = \underline{\boldsymbol{x}}_{\bot}^T\underline{\boldsymbol{x}}_{\bot}$.
\end{lemma}

Then we state the main stability result of this paper:

\begin{theorem}\label{thm: stability_complete_system}
For the system in (\ref{eq: ctrl affine_noise}) with the proposed control law shown in Fig.\ref{Complete closed-loop system} if the following conditions are met:
\begin{enumerate}
    \item applying the baseline controller with CCM in (\ref{saturated control signal feedback}), 
    \item applying the UDE in (\ref{eq: hat_delta_j}) and (\ref{eq: delta_T_allocation}),
    \item applying  the AGAS tracker in (\ref{eq: attitude_control}),
    \item assumption \ref{assm: disturbance_assm} is satisfied.
\end{enumerate}
then all trajectories of the closed-loop system $\underline{\boldsymbol{\eta}}$  converge to the reference trajectory $\underline{\boldsymbol{\eta}}_0$, i.e. $||\underline{\boldsymbol{\eta}}_1 - \underline{\boldsymbol{\eta}}_0|| \rightarrow 0$ as $t\rightarrow\infty$. In addition, the control force applied to the system is bounded such that $||\underline{\boldsymbol{\zeta}}||\leq \zeta_b$.

\end{theorem}

% \begin{theorem}\label{th: ude_convergence}
% If the conditions in Assumption \ref{assm: disturbance_assm} are met, the estimation errors $\underline{\boldsymbol{\tilde{\delta}}}_T$ and $\underline{\boldsymbol{\tilde{\delta}}}_{\bot,j}$ are bounded. In addition, if the contraction condition in (\ref{contraction condition}) and (\ref{contraction condition bound}) are met, then $\underline{\boldsymbol{\tilde{\delta}}}_T  \rightarrow \underline{\boldsymbol{0}}$ and $\underline{\boldsymbol{\tilde{\delta}}}_{\bot,j}  \rightarrow \underline{\boldsymbol{0}}$ as $t \rightarrow \infty$.
% \end{theorem}
\begin{proof}
First, we analyze the properties of the UDE. A Lyapunov function $V_e$ is defined as follows: 
\begin{equation} \label{eq: V_e_definition}
\begin{aligned}
& V_e = \frac{1}{2}c_T\underline{\boldsymbol{\tilde{\delta}}}_T^T\underline{\boldsymbol{\tilde{\delta}}}_T +  \frac{1}{2}\sum^{N}_{j=1}\Big[ c_T\lambda_T N /(2\kappa_{j})  +   c_j/N  \Big ] \underline{\boldsymbol{\tilde{\delta}}}_j^T\underline{\boldsymbol{\tilde{\delta}}}_j
\end{aligned}
\end{equation}
where $ c_T$, $c_j$ are positive constants. According to the error dynamics in (\ref{eq: delta_j_err_dynamics}) and (\ref{eq: dynamics_delta_T}), the time derivative of $V_e$ is:
\begin{equation}
\begin{aligned}
& \dot{V}_e =  - c_T \lambda_{T} \underline{\boldsymbol{\tilde{\delta}}}_T^T \underline{\boldsymbol{\tilde{\delta}}}_T- c_T \lambda_{T} \sum^{N}_{j=1} \underline{\boldsymbol{\tilde{\delta}}}_{\bot,j}^T\underline{\boldsymbol{\tilde{\delta}}}_T \\ 
& -  \sum^{N}_{j=1}\Big[ c_T\lambda_T N /2  + c_j\kappa_{j}/N  \Big] \underline{\boldsymbol{\tilde{\delta}}}_j^T \mathfrak{B}_j\underline{\boldsymbol{\tilde{\delta}}}_{\bot,j}  
\end{aligned}
\end{equation}
According to (\ref{eq: B_j_column}), $\mathfrak{B}_j \underline{\boldsymbol{l}}_j = \underline{\boldsymbol{0}}$, we have $\underline{\boldsymbol{\tilde{\delta}}}_j^T \mathfrak{B}_j = \underline{\boldsymbol{\tilde{\delta}}}_{\bot,j}^T \mathfrak{B}_j$. Using Lemma \ref{lem: B_p}, we can obtain $\underline{\boldsymbol{\tilde{\delta}}}_j^T \mathfrak{B}_j\underline{\boldsymbol{\tilde{\delta}}}_{\bot,j} = \underline{\boldsymbol{\tilde{\delta}}}_{\bot,j}^T \underline{\boldsymbol{\tilde{\delta}}}_{\bot,j}$. Hence, $ \dot{V}_e$ is:
\begin{equation}\label{eq: v_e_dot}
\begin{aligned}
\begin{aligned}
& \dot{V}_e = - \sum^{N}_{j=1} 
\underline{\boldsymbol{\tilde{z}}}_j^T \begin{bmatrix}
c_T \lambda_T /N \cdot \boldsymbol{1} & c_T  \lambda_T/2  \cdot \boldsymbol{1}\\
c_T  \lambda_T/2  \cdot \boldsymbol{1} & c_T \lambda_T N/2 + c_j\kappa_j/N  \cdot \boldsymbol{1} \\
\end{bmatrix} \underline{\boldsymbol{\tilde{z}}}_j
\end{aligned}
\end{aligned}
\end{equation}
where $\underline{\boldsymbol{\tilde{z}}}_j =[ \underline{\boldsymbol{\tilde{\delta}}}_T^T, \quad\underline{\boldsymbol{\tilde{\delta}}}_{\bot,j}^T]^T $. It is trivial to verify that $\dot{V}_e$ is negative semi-definite. Note that since $V_e$ is a positive definite Lyapunov function, we can conclude that disturbance estimation errors $\underline{\boldsymbol{\tilde{\delta}}}_T$, $\underline{\boldsymbol{\tilde{\delta}}}_{j}$, and $\underline{\boldsymbol{\tilde{\delta}}}_{\bot,j}$ are bounded. According to Assumption \ref{assm: disturbance_assm}, $\underline{\boldsymbol{\hat{\delta}}}_T$ and $\underline{\boldsymbol{\hat{\delta}}}_{\bot,j}$ are bounded. By Theorem \ref{thm: contraction_robustness}, the trajectory errors are bounded by external disturbances $\underline{\boldsymbol{\delta}}_e$ and lift force error $\underline{\boldsymbol{\zeta}}_e$. With the application of (\ref{eq: hat_delta_j}) and (\ref{eq: delta_T_allocation}) the AGAS attitude tracker in (\ref{eq: attitude_control}), $\underline{\boldsymbol{d}}$ is as follows:
\begin{equation}
    \underline{\boldsymbol{d}} =  \boldsymbol{G}(\underline{\boldsymbol{x}}) \underline{\boldsymbol{\zeta}}_e + \boldsymbol{G}_{\delta}(\underline{\boldsymbol{x}}) \underline{\boldsymbol{\delta}}_e
\end{equation}
where $\underline{\boldsymbol{\delta}}_e = - [\underline{\boldsymbol{\tilde{\delta}}}_{T}^T , \, \underline{\boldsymbol{\tilde{\delta}}}_{\bot,1}^T, \, \underline{\boldsymbol{\tilde{\delta}}}_{\bot,2}^T, \, \underline{\boldsymbol{\tilde{\delta}}}_{\bot,3}^T]^T$ (see Sec. 4 of support document\footnotemark[1] for details). Since  $\underline{\boldsymbol{\hat{\delta}}}_T$ and $\underline{\boldsymbol{\hat{\delta}}}_{\bot,j}$ are bounded, and AGAS attitude tracker is used, $\underline{\boldsymbol{d}}$ is bounded, and $ ||\underline{\boldsymbol{\eta}}_1 - \underline{\boldsymbol{\eta}}_0||$ is bounded. Hence, the state of the closed-loop system $\underline{\boldsymbol{x}}$ is bounded. In addition, by using the dynamics of the estimation error in (\ref{eq: delta_j_err_dynamics}) and (\ref{eq: dynamics_delta_T}) together with $\underline{\boldsymbol{x}}$ being bounded, we conclude that $\dot{\underline{\boldsymbol{\tilde{\delta}}}}_T$ and $\dot{\underline{\boldsymbol{\tilde{\delta}}}}_j$ are bounded. Hence,  $\ddot{V}_e$ is bounded and $\dot{V}_e$ is uniformly continuous (see Sec. 5 of support document\footnotemark[1] for details). According to Barbalat's Lemma, $\dot{V}_e \rightarrow 0$ as $t \rightarrow \infty$, and we conclude that $\underline{\boldsymbol{\tilde{\delta}}}_T  \rightarrow \underline{\boldsymbol{0}}$ and $\underline{\boldsymbol{\tilde{\delta}}}_{\bot,j}  \rightarrow \underline{\boldsymbol{0}}$ as $t \rightarrow \infty$. Finally, $ \underline{\boldsymbol{d}} \rightarrow \underline{\boldsymbol{0}} $ as $t \rightarrow \infty$.
 Hence, according to Theorem \ref{thm: contraction_robustness}, $ ||\underline{\boldsymbol{\eta}}_1 - \underline{\boldsymbol{\eta}}_0|| \rightarrow 0$ as $t \rightarrow \infty$.

 Moreover, the magnitude of the control force is bounded as $||\underline{\boldsymbol{\zeta}}|| = ||\underline{\boldsymbol{\zeta}}_c|| = ||\underline{\boldsymbol{\zeta}}_{nn,sat} + \underline{\boldsymbol{f}}_\delta|| \leq f_b + ||\underline{\boldsymbol{k}}^*|| + ||\underline{\boldsymbol{f}}_\delta|| \leq \zeta_b$

% \begin{theorem}\label{thm: stability_complete_system}
% For the slung-load system in (\ref{eq: ctrl affine_noise}), if we apply the control force obtained from the contraction analysis in (\ref{eq: control_saturation}), the disturbance compensation in (\ref{eq: hat_delta_j}) and (\ref{eq: delta_T_allocation}), and the AGAS tracker in (\ref{eq: attitude_control}), then all trajectories of the closed-loop system converges to the reference trajectory exponentially, i.e. $||\underline{\boldsymbol{\eta}}_1 - \underline{\boldsymbol{\eta}}_0|| \rightarrow 0$ as $t\rightarrow\infty$.
% \end{theorem}

\end{proof}
\section{Simulation Verification}\label{sec: simulation_verification}
Our model contains a $1.3$kg point mass payload attached to three drones with inelastic $0.98$m cables; each drone is $1.5$kg. For the reference trajectory, each cable should form a $30^{\circ}$ horizontal angle and a $15^{\circ}$ vertical angle ($\theta_{xy}$, $\theta_z$) with respect to its projection ($\underline{\boldsymbol{r}}_j$) according to Fig.\ref{fig: system_config}. 

\begin{figure}
    \centering
    \begin{subfigure}{0.48\textwidth}
        \centering
        \includegraphics[width=\textwidth]{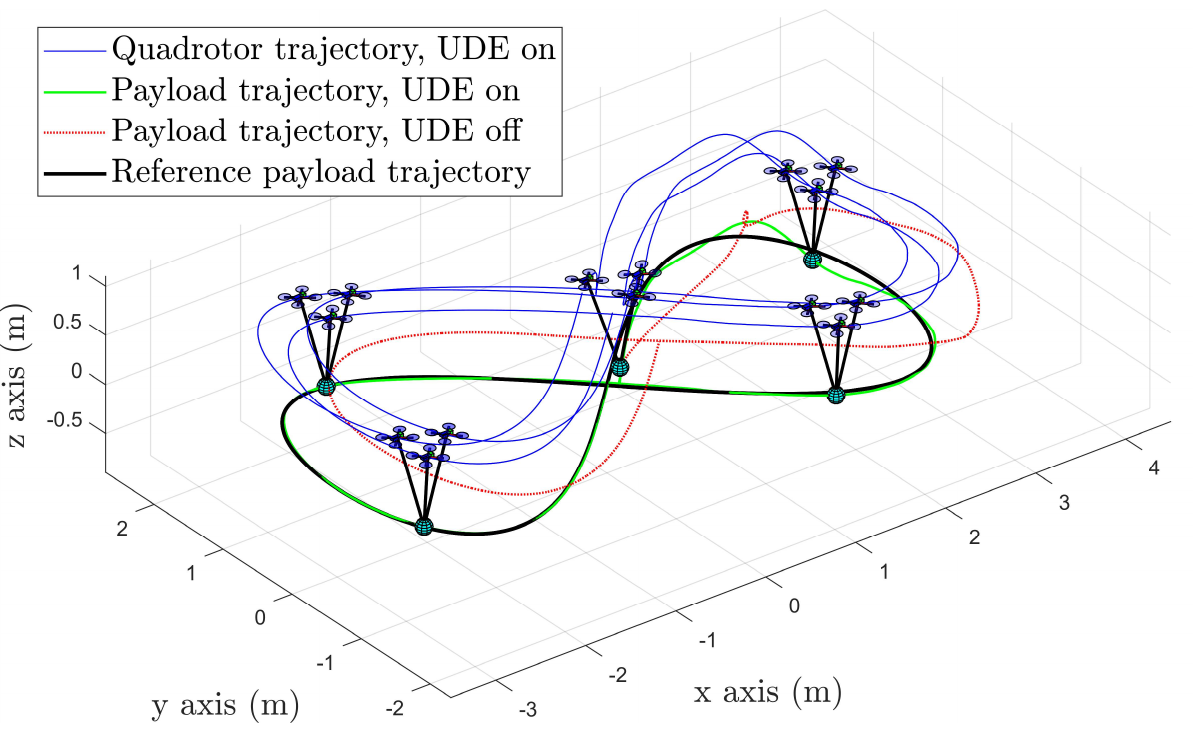}
        \caption{Trajectory tracking results.}
        \label{Figure-8 trajectory}
    \end{subfigure}
    \begin{subfigure}{0.48\textwidth}
        \centering
        \includegraphics[width=\textwidth]{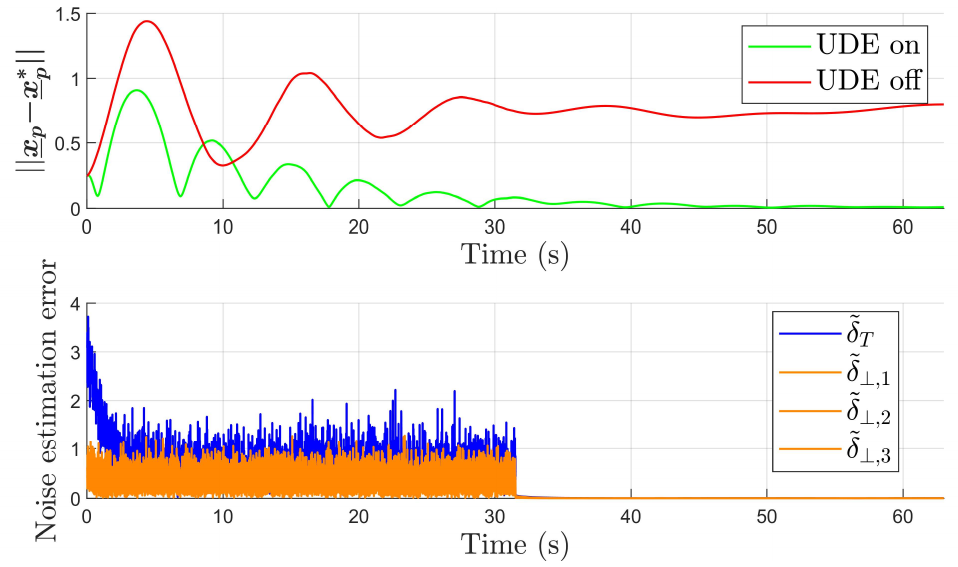}
        \caption{Payload trajectory tracking errors and UDE estimation errors.}
        \label{Figure-8 path UDE error}
    \end{subfigure}
    \caption{Trajectory tracking MATLAB simulation performance plots, (a): comparison with UDE on and off; (b): payload tracking error and noise estimation error with respect to time when using UDE.}
    \label{Figure-8 path MATLAB and Python simulation performance plots}
    \vspace{-0.3cm}
\end{figure}

\subsection{The Training Environment Setup}
We deploy fully connected neural networks for the dual metric $\boldsymbol{W}$ and the controller $\underline{\boldsymbol{k}}$ with randomly sampled datasets. All neural networks have 2 layers with 128 neurons in the hidden layer. The training was executed on the Flight Systems and Control Lab (FSC Lab) server, which is equipped with an RTX 4060 GPU and an Intel i5 CPU.

\subsection{Trajectory Tracking Under External Disturbances}
The performance of the figure-8 trajectory tracking is demonstrated in Fig.\ref{Figure-8 path MATLAB and Python simulation performance plots}. The disturbance force is a summation of constant and stochastic noise $\underline{\boldsymbol{\delta}} = \underline{\boldsymbol{\delta}}_c + \underline{\boldsymbol{\delta}}_s$, where $\underline{\boldsymbol{\delta}}_c = [0.3, -0.2, 0.5, 0.3, ... , 0.3]^T \in \mathbb{R}^{12\times1}$ and $\underline{\boldsymbol{\delta}}_s \sim 0.3\cdot\mathcal{U}(0,1)$ is uniformly distributed. The control bound $f_b$ is set at $3$ with the saturation factor $a=0.3$. The simulation lasts for $63$ seconds and the stochastic noise is set to 0 (only constant noise after this) at $t = 31.5s$. Accuracy is significantly improved with the UDE turned on. Even with Assumption \ref{assm: disturbance_assm} not satisfied, the noise estimation and payload trajectory can quickly converge to a bounded neighbourhood of the reference. After the stochastic noise is turned off, Assumption \ref{assm: disturbance_assm} is fully satisfied. The noise estimation error and payload tracking error converge to $0$, confirming the stability analysis of Theorem \ref{thm: stability_complete_system}. Therefore, our proposed control law can fulfill slung payload trajectory tracking under input saturation and external disturbances. Additional simulation results and the source codes are available in our GitHub repository\footnotemark[1].  

\section{CONCLUSIONS}\label{sec: con}
In this paper, we present a neural CCM design for robust multi-drone slung payload transportation systems. An extensive derivation of the dynamics, contraction metric, and disturbance estimation is provided. Stability and robustness are proved, with results illustrated by numerical simulations. Future work will focus on physical experiments and state constraints of the contraction metric.

% \bibliography{references}

\addtolength{\textheight}{-12cm}   % This command serves to balance the column lengths

%%%%%%%%%%%%%%%%%%%%%%%%%%%%%%%%%%%%%%%%%%%%%%%%%%%%%%%%%%%%%%%%%%%%%%%%%%%%%%%%

%%%%%%%%%%%%%%%%%%%%%%%%%%%%%%%%%%%%%%%%%%%%%%%%%%%%%%%%%%%%%%%%%%%%%%%%%%%%%%%%


% \section*{ACKNOWLEDGMENT}

% The preferred spelling of the word ÒacknowledgmentÓ in America is without an ÒeÓ after the ÒgÓ. Avoid the stilted expression, ÒOne of us (R. B. G.) thanks . . .Ó  Instead, try ÒR. B. G. thanksÓ. Put sponsor acknowledgments in the unnumbered footnote on the first page.

\end{document}